\documentclass[twocolumn,prb,superscriptaddress,floatfix,letterpaper]{revtex4-1}
\pdfoutput=1
\usepackage{amsmath,amssymb,amsthm} 
\usepackage[colorlinks,citecolor=blue]{hyperref}

\DeclareMathOperator{\Tr}{Tr}

\DeclareMathOperator{\Hom}{Hom}
\DeclareMathOperator{\Rep}{Rep}

\newcommand{\xt}{\times}
\newcommand{\ot}{\otimes}

\newcommand{\Ab}[1]{{#1}_{Ab}}
\newcommand{\GAb}[1]{\underline{{#1}_{Ab}}}

\newcommand{\one}{\mathbf{1}}

\newcommand{\Z}{\mathbb{Z}}

\newcommand{\Q}{\mathbb{Q}}

\newcommand{\bl}{\boldsymbol} 
 
\newcommand{\ii}{\mathrm{i}}
\newcommand{\ee}{\mathrm{e}}

\newcommand{\<}{\langle} 
\renewcommand{\>}{\rangle} 

\newcommand{\Ref}[1]{Ref.~\onlinecite{#1}}

\newcommand{\sgn}{{\rm sgn}}

\newcommand{\cC}{ {\cal C} } 
\newcommand{\cD}{ {\cal D} }

\newcommand{\bpm}{\begin{pmatrix}}
\newcommand{\epm}{\end{pmatrix}}

\theoremstyle{definition}
\newtheorem{thm}{Theorem}

\newtheorem{conj}{Conjecture}

\newcommand{\ov}{\overline}

\newcommand{\Fun}{\mathrm{Fun}}

\newcommand{\prlsec}[1]{\medskip \noindent\textbf{#1:}}

\begin{document}

\title{Matrix formulation for non-Abelian families}

\author{Tian Lan} 
\affiliation{Institute for Quantum Computing,
  University of Waterloo, Waterloo, Ontario N2L 3G1, Canada}
\affiliation{Center for Quantum Computing, Peng Cheng Laboratory,
  Shenzhen 518055, China}
  \begin{abstract}
    We generalize the $K$ matrix formulation to non-trivial non-Abelian families of 2+1D topological orders. Given a topological order $\mathcal C$, any topological order in the same non-Abelian family as $\mathcal C$ can be efficiently described by $\boldsymbol a=(a_I)$ where $a_I$ are Abelian anyons in $\mathcal C$, together with a symmetric invertible matrix $K$, $K_{IJ}=k_{IJ}-t_{a_I,a_J}$ where $k_{IJ}$ are integers, $k_{II}$ are even and $t_{a_I,a_J}$ are the mutual statistics between $a_I,a_J$. In particular, when $\mathcal C$ is a root whose rank is the smallest in the family, $K$ becomes an integer matrix. Our results make it possible to generate the data of large numbers of topological orders instantly.
  \end{abstract}

\maketitle

\prlsec{Introduction}
Topological phases of matter have drawn more and more research interest during
recent years. A most remarkable feature of topological phases is that there can
be several quantum states which are ``topologically'' degenerate. Such
degeneracy is robust against any local perturbation, thus these states can be
employed as qubits that are automatically immune to local noises. Given the
possible application in quantum memory and quantum computation, it is then
natural to ask how to produce the desired topological degeneracy.

One source of topological degeneracy is to put topological ordered system on a
manifold with nontrivial topology~\cite{Wen89,Wen90,WN90,Kitaev9707021}. This approach is not
ideal: for one reason, it is not easy to shape a physical system into nontrivial
manifold such as a torus; for another, to manipulate the degenerate ground
states one has to perform non-local operations.

Another source of topological degeneracy is to trap several anyonic
quasiparticles. By braiding and fusion of these anyons, it is
possible to realize universal topological quantum computation~\cite{FKLW01}. For
an anyon $i$, we
use the quantity $d_i$, called the quantum dimension, to measure the effective
topological degeneracy carried by $i$. When there is a large number $N$ of anyon
$i$ trapped, the topological degeneracy is of the order $d_i^N$.

Thus for anyons to produce desired topological degeneracy, it is necessary that
$d_i>1$. An anyon with $d_i=1$ is called Abelian while with $d_i>1$ is called
non-Abelian. If all the anyons in a topological order are Abelian, it is called
an Abelian topological order. Clearly Abelian topological phases are useless in
the braiding-fusion based topological quantum computation.

In \Ref{LW1701.07820} we proposed the generalized hierarchy construction that can add or
remove Abelian anyons to or from any topological order. Two topological orders
which can be connected by such construction are of the same ``non-Abelian
family'', which is the equivalence class up to Abelian topological orders. The
non-Abelian family captures the invariants of non-Abelian anyons, and we expect
that topological orders in the same non-Abelian family behave similarly in
topological quantum computation.

However, the construction in \Ref{LW1701.07820} is performed in a step-by-step
manner. Given a topological order $\cC$, it is not easy to calculate the
property of another
topological order in the same non-Abelian family that requires several steps of
hierarchy constructions from $\cC$. This letter aims at resolving such difficulty. 
We showed that given a topological order $\cC$, any topological order in the
same non-Abelian family can be efficiently represented by a sequence of Abelian
anyons in $\cC$ together with a $K$ matrix. When $\cC$ is the trivial
topological order, our result reduces to the original $K$ matrix formulation for
Abelian topological orders~\cite{WZ92}.

\prlsec{One-step generalized hierarchy construction}
We first review and refine the construction proposed in \Ref{LW1701.07820}.
The main idea is to let Abelian anyons
form an effective Laughlin-like state~\cite{Laughlin1983}. This
idea dates back to Haldane and Halperin, known as ``hierarchy''
construction~\cite{Haldane83,Halperin84}.
But below we discuss it at a more general level.

We start with a topological phase $\cC$. The anyons in $\cC$ are labeled by
$i,j,k,\cdots$.  Let $a_c$ be an Abelian anyon in $\cC$ with topological spin
$s_{a_c}$. $s_{a_c}$ determines the self statistics of $a_c$: exchanging two
$a_c$ anyons leads to the phase factor $\ee^{2\pi\ii s_{a_c}}$.  We
try to make $a_c$ form the Laughlin state,
\begin{align}
  \langle\{z_a\}|\Psi\rangle=\prod_{a<b}(z_a-z_b)^{M_c}\times\ee^{-\frac14\sum|z_a|^2}.
\end{align}
The resulting topological
phase is determined by $\cC$, $a_c$ and $M_c$, which will be denoted by
$\cC_{a_c,M_c}$. 
    Here $z_a,z_b$ are the positions of $a_c$ anyons. $M_c$ must be consistent with
    anyon statistics.
    Consider
    exchanging two $a_c$ anyons, we obtain: a phase factor $\ee^{2\pi\ii
      \frac{M_c}{2}}$ from the wave function and a phase factor $\ee^{2\pi\ii s_{a_c}}$ from anyonic statistics.
     To be consistent, total phase factor must be 1:
     \begin{align}
     \frac{M_c}{2}+s_{a_c}\in \Z.
   \end{align}
	  So we need to take {$M_c=m_c-2s_{a_c}$}, where
	  $m_c$ is an even
	integer.

  Anyon $i$ in the phase $\cC$ may be dressed with a flux $M_i$ in the new phase
  ${\cC_{a_c,M_c}}$.
  \begin{align}
    \Psi(i,M_i)=\prod_b(\xi_i-z_b)^{M_i}\prod_{a<b}(z_a-z_b)^{M_c}\times\ee^{-\frac14\sum|z_a|^2}.
  \end{align}
Here $\xi_i$ is the position of anyon $i$.
Thus an anyon in the new phase is represented by a
pair $(i,M_i)$.
Again, $M_i$ can not be arbitrary. If $a_c$ has trivial mutual statistics with
$i$, $M_i$ can be any integer. Otherwise, consider moving $a_c$ around
$(i,M_i)$ and we obtain: a phase factor $\ee^{2\pi\ii M_i}$ from the flux $M_i$
and a phase factor $\ee^{2\pi\ii {t_i}}$ from the mutual statistics between
    $a_c$ and $i$. The mutual statistics can be extracted from the $S$ matrix,
    $\ee^{2\pi\ii {t_i}}=DS_{i a_c^*}/d_i$, $t_{a_c}=2s_{a_c}$.
To be consistent, total phase factor must be 1:
\begin{align}
  M_i+t_i\in\Z.
\end{align}

  Since the anyon $a_c$ dressed with a flux $M_c$ is a ``trivial excitation'' in the
  new phase:
  \begin{gather}
    \begin{align}
      \Psi(a_c,M_c)&\sim\prod_b^n(\xi_{a_c}-z_b)^{M_c}\prod_{a<b}^n(z_a-z_b)^{M_c}\nonumber\\&=\prod_{a<b}^{n+1}(z_a-z_b)^{M_c},\nonumber\\
     & (a_c,M_c)\sim (\one,0),
    \end{align}
  \end{gather}
 we have the equivalence relation: 
  \begin{align}
    (i,M_i)\sim (i\otimes a_c,M_i+M_c). \label{equivim}
    \end{align}

  Next we list the data of the resulting topological order ${\cC_{a_c,M_c}}$:
  \begin{itemize}
    \item 
      The spin of $(i,M_i)$ is given by the spin of $i$ plus the ``spin'' of  the flux $M_i$:
      \begin{align}
	s_{(i,M_i)}=s_i+\frac{M_i^2}{2M_c}.\label{spinim}
      \end{align}
    \item 
      To fuse anyons $(i,M_i),(j,M_j)$ in the new phase, just fuse $i,j$ as in
      the old phase, and add up the flux:
      \begin{align}
	(i,M_i)\otimes (j,M_j)=\bigoplus_k N^{ij}_k (k,M_i+M_j).
      \end{align}
      And then apply the equivalence relation \eqref{equivim}.
    \item 
      The rank (number of anyon types) of
      ${\cC_{a_c,M_c}}$ is
      \begin{align}
	N^{\cC_{a_c,M_c}}=|M_c|N^\cC.
      \end{align}
    \item 
      The quantum dimensions remain the same
      \begin{align}
	 d_{(i,m)}=d_i.
      \end{align}
    \item 
      The $S$ matrix is
      \begin{align}
	\label{Sd}
	S^{{\cC_{a_c,M_c}}}_{(i,M_i),(j,M_j)}&=\frac{1}{\sqrt{|M_c|}}S^{\cC}_{ij}\ee^{-2\pi\ii\frac{M_iM_j}{M_c}}.
      \end{align}
    \item 
      The chiral central charge is
      \begin{align}
	c^{\cC_{a_c,M_c}}=c^\cC+\sgn M_c .\label{diffc}
      \end{align}
  \end{itemize}

The one-step hierarchy construction is reversible. In
${\cC_{a_c,M_c}}$, choosing $a_c'=(\one,1),\ s_{a_c'}=\frac{1}{2M_c},\ m_c'=0,\ M_c'=-1/M_c$, and repeating
the construction, we will go back to $\cC$.  
Therefore, hierarchy construction defines an
equivalence relation between topological phases. We call the corresponding
equivalence classes the ``non-Abelian families''.
Each non-Abelian family have ``root'' phases with
the smallest rank. Let $\Ab{\cC}$ denote the full subcategory of all Abelian
anyons in $\cC$, $\cC$ is a root if~\cite{Lan2017} and only
if~\cite{LW1701.07820,Lan2017} $\Ab{\cC}$ is a symmetric fusion
category, namely all the Abelian anyons are bosons or fermions with trivial
mutual statistics with each other.

\prlsec{Multiple steps of construction and the matrix formulation}
Now we consider multiple steps of hierarchy constructions and try to write down
the final result at once. 
Note that in the flux label $M_i$ we need to use the mutual statistics in the
previous step, and things get involved when there are multiple steps. To
separate out the mutual statistics and thus make things clearer, we use the ``integer convention''
$(i,m)$,
instead of the ``flux convention'' $(i,M_i)$, where $m-t_i=M_i$.

Now consider starting from a topological order $\cC$ and performing one-step construction
$\kappa$ times. The first step we take $a_1\in \cC_{Ab}$ and even integer
$k_{11}$. The second step we take an Abelian anyon $(a_2\in\cC_{Ab},k_{12})$
and even integer $k_{22}$, where $k_{12}$ is an integer.
The third step we take an Abelian anyon
$((a_3\in\cC_{Ab},k_{13}),k_{23})$
and even integer $k_{33}$, where $k_{13},k_{23}$ are integers.
Keep moving on
and we see that the steps can be summarized by $a_I$ and $k_{IJ}$. Define a
corresponding integer symmetric $\kappa$ by $\kappa$
matrix by setting $k_{IJ}=k_{JI}$.
Denote by $t_{i,a}$ the mutual statistics
between anyon $i$ and Abelian anyon $a$ in $\cC$ ( $\ee^{2\pi\ii
t_{i,a}}$ is the phase factor of braiding $a$ around $i$), by $s_i$ the
spin of anyon $i$ in $\cC$, and set $t_{a,a}=2s_a$.
Let the $K$ matrix be
$K_{IJ}=k_{IJ}-t_{a_I,a_J}$.

Physically, we let the Abelian anyons $a_I,\ I=1,2,\dots,\kappa$ form a
multilayer Laughlin-like state
\begin{align}
  \prod (z^{(I)}_a-z^{(J)}_b)^{K_{IJ}},
\end{align}
where $I$ labels the layer and $z^{(I)}_a$ is the position of $a_I$ anyon.  By a
similar argument as in the one-step case, we know that $K_{IJ}+t_{a_I,a_J}$ must
be an integer and $K_{II}+t_{a_I,a_I}$ must be an even integer. 

Though we are using the integer convention, note that similar to the one-step case, it is the combination
$K_{IJ}=k_{IJ}-t_{a_I,a_J}$ or $M_c=m_c-t_{a_c}$ that determines the final topological
order, not the integer $k_{IJ}$ or $m_c$ alone. The meaning of $k_{IJ}$ or $m_c$
depends on the choice of mutual statistics $t_{i,a}$.

The fusion rule and $T,S$ matrices of the resulting topological order after
$\kappa$ steps can be calculated efficiently via the $K$ matrix as stated in
Theorem \ref{thm.main}. This
result generalizes the $K$ matrix formulation for Abelian topological
orders~\cite{WZ92}.

\begin{thm}\label{thm.main}
  The topological order constructed from root $\cC$ via $\kappa$ steps
  can be summarized by $a_I$ and $K_{IJ}$, where $I,J=1,\dots,\kappa$,
  $a_I\in\cC_{Ab}$, $\det K\neq 0$, $K_{IJ}=K_{JI}$, $K_{IJ}+t_{a_I,a_J}$ are integers
  and $K_{II}+t_{a_I,a_I}$ are even. Let $\bl a$ formally denote the vector $(a_I)$ and
  $\cC_{\bl a,K}$ denote the resulting topological order. $\cC_{\bl a,K}$ is as
  follows:
  \begin{itemize}
    \item Fix a choice of mutual statistics $t_{i,a_I}$ in $\cC$.
    Let $\bl t_{i}$ be the $\kappa$-dimensional vector
      $(t_{i,a_I})$. Anyons are labeled by $(i\in\cC,\bl l)$ where $\bl l$ is a
      $\kappa$-dimensional integer vector, subject to the following equivalence
      relations
      \begin{align}
      \label{ilequiv}
	(i,\bl l)\sim (i\ot a_I,\bl l +K_I-\bl t_i+\bl t_{i\ot a_I}).
      \end{align}
      where $K_I$ is the $I$th column vector of $K$.
      For a different choice of mutual statistics, or
      representative $i'\sim i$, $t_{i',a_I}$ differs from $t_{i,a_I}$ by an integer,
      and $(i',\bl l+\bl t_{i'}-\bl t_{i})\sim (i,\bl l)$.
      $\cC_{\bl a,K}$ does not depend on the choice of mutual statistics or representative in $\cC$.
      
      \item Fusion is given by 
      \begin{align}
      \label{ilfuse}
	(i,\bl l)\ot (j,\bl k)=\oplus N^{ij}_s(s,\bl l+\bl k-\bl t_i-\bl t_j
	+\bl t_s).
      \end{align}
    \item 
      The spin of $(i,\bl l)$ is
      \begin{align}
      \label{ilspin}
	s_{(i,\bl l)}=s_i+\frac12 (\bl l-\bl t_i)^T K^{-1} (\bl l-\bl t_i). 
      \end{align}
    \item The $S$ matrix is
      \begin{align}
      \label{ilmutual}
	S_{(i,\bl l)(j,\bl k)}=\frac 1{\sqrt{|\det{K}|}}S_{ij}\ee^{-2\pi\ii
	(\bl l-\bl t_i)^T K^{-1}(\bl k -\bl t_j)}.
      \end{align}
      \item The rank is $N^{\cC_{\bl a,K}}=|\det K| N^\cC$. The chiral central charge is $c^{\cC_{\bl a,K}}=c^\cC+\sgn K.$ Here $\sgn K$ denotes the index of the matrix $K$, namely the number of positive eigenvalues minus the number of negative eigenvalues.
  \end{itemize}
\end{thm}
\begin{proof}
  We postpone the lengthy proof to Appendix~\ref{ap.proof}.
\end{proof}
 When $\cC$ is a root whose Abelian anyons
 $\cC_{Ab}$ is a symmetric fusion category,
$a_I,a_J$ are mutually trivial, and $t_{a_I,a_J}$ are all integers. In particular,
we can choose $t_{a_I,a_I}=1$ when $a_I$ is fermionic, and other
$t_{a_I,a_J}=0$.  
In this case the $K$ matrix is an integer matrix and $K_{II}$ is even
 when $a_I$ is a boson and odd when
 $a_I$ is a fermion.

\prlsec{Equivalence relation of $\cC_{\bl a,K}$} Starting form the same topological order $\cC$, different paths of
  construction may result in the same topological order. It is natural to ask
  what is the
  equivalence relation for $(\bl a,K)$. For now, we know three ways to generate
  equivalent $\cC_{\bl a,K}$:

  \begin{enumerate}
  \item The equivalence between the starting point $F:\cC\simeq \cD$ naturally
    give rise to equivalence $\cC_{\bl a,K}\simeq \cD_{F(\bl a),K} $.
  \item ``Integer linear recombination'' of $a_I$, $W\in GL(\kappa,\Z)$
    (namely $W$ is
    an integer matrix with $\det W=\pm 1$),
    $\cC_{\bl a,K}\simeq \cC_{W\bl a, WKW^T}$. We call such transformation as
    the $GL(\Z)$ transformation.
  \item 
  The reversibility of one-step construction means that the topological order
  constructed from
  $\cC$ with $\bpm a_1=a_c\\  a_2=\one\epm,$ $K=\bpm M_c& 1\\1&0\epm$ is equivalent to
  $\cC$. Also $(a_I, K_{IJ})$ is equivalent to $\bpm a_I\\ a\\ \one\epm, \bpm K_{IJ} &
  \bl l_c-\bl t_{a} & 0\\ \bl l_c^T-\bl t_{a}^T & m_c-2s_{a} & 1\\0&1&0\epm$,
  where $a$ can be any Abelian anyon in $\Ab{\cC}$.
  Note that under $GL(\Z)$ transformation, $\bpm K_{IJ} &
  \bl l_c-\bl t_{a} & 0\\ \bl l_c^T-\bl t_{a}^T & m_c-2s_{a} & 1\\0&1&0\epm\sim
  \bpm K_{IJ} &0&0 \\ 0 &-2s_a &1\\ 0&1&0\epm=K \oplus \bpm
  -t_{a,a}&1\\1&0\epm$. Therefore, we have $({\bl a,K})\sim \left(\bl a\oplus
 \bpm a \\ \one\epm,K\oplus \bpm -t_{a,a}&1\\1&0\epm\right)$. We refer to
 $\left( \bpm a\\ \one\epm,\bpm -t_{a,a}&1\\1&0\epm
  \right)$ as the ``trivial bilayer''.
  \end{enumerate}
\begin{conj}
  $\cC_{\bl a,K}$ and $\cC_{\bl a',K'}$ (with exactly the same chiral central
  charge, not modulo 8) are equivalent if and only if,
  up to automorphisms of $\cC$ and $GL(\Z)$ transformations, $(\bl a\oplus \bl
  b, K\oplus X)\sim (\bl a' \oplus \bl b', K'\oplus X')$ where $(\bl b,X)$ and
  $(\bl b',X')$ are direct sums of trivial bilayers $\left( \bpm a\\ \one\epm,\bpm -t_{a,a}&1\\1&0\epm
  \right)$.
\end{conj}

\prlsec{The formal categorical formulation}
We give the formal basis independent formulation of the above
constructions.
Let $\cC$ be a braided fusion category, $\alpha_{A,B,C}, c_{A,B}$ denote the
associator and braiding in $\cC$, $\GAb{\cC}$ denote the Abelian group
corresponding to the pointed subcategory $\Ab{\cC}$, and $t: \mathrm{Irr}(\cC) \xt
\GAb{\cC}\to \Q$ denote the mutual statistics between simple objects and
pointed ones, namely $\ee^{2\pi\ii
t(i,a)}=\frac 1{d_i}\Tr c_{a,i}c_{i,a}$; in particular, the diagonal entries
are related to exchange statistics $\ee^{\ii\pi t(a,a)}=\Tr c_{a,a}$.

Let $\Z^\kappa$ be a free Abelian group with $\kappa$ generators. It can be
naturally extended  to a $\kappa$ dimensional vector space over $\Q$. Let
$\ov\Z^\kappa:=\Hom(\Z^\kappa, \Q)$ denote the ``dual space'', the space of
$\Q$-linear
functions. Conventionally, we use $x,y,\dots$ to denote elements in $\Z^\kappa$
and $f(-),g(-),\dots$, or simply $f,g$ when not confusing, to denote functions in $\ov\Z^\kappa$.

Let $K:\Z^\kappa\xt\Z^\kappa\to \Q$ be a non-degenerate symmetric bilinear form.
It defines an isomorphism from $\Z^\kappa$ to $\ov\Z^\kappa$, by $x \mapsto
K(x,-)=K(-,x)$. Denote the inverse map by $\tilde K$, thus
\begin{align}
  \tilde K(K(x,-))=x, \quad K(\tilde K(f),x)=f(x).
\end{align}
There is then a natural non-degenerate symmetric bilinear form $\ov K$ on $\ov\Z^\kappa$
induced from $K$, via
\begin{align}
  \ov K( f,g)=K(\tilde K(f),\tilde K(g))=f(\tilde K(g))=g(\tilde K(f)).
\end{align}
If one chooses a basis of $\Z^\kappa$ and the corresponding dual basis of
$\ov\Z^\kappa$, the matrix of $K$ and $\ov K$ are inverse to each other.

We also need to choose $\kappa$ Abelian anyons for each step. This is concluded
in a group homomorphism $\bl a: \Z^\kappa\to \GAb{\cC}$. The bilinear form $K$
needs to satisfy the even integral condition, namely $\forall x,y$, $K(x,y)+t(\bl a(x),\bl a(y))\in
\Z$ and $K(x,x)+t(\bl a(x),\bl a(x))\in 2\Z$.

For a $\kappa$ step construction, we first define a semisimple category
$\cC^\uparrow_{\bl a,K}$. $\cC^\uparrow_{\bl a,K}$ is graded by
$\ov\Z^\kappa/K(2\ker \bl a,-)$ (not
faithful). Take a representative $f\in \ov\Z^\kappa$, the component
$(\cC^\uparrow_{\bl a,K})_{f}$ is a full subcategory of $\cC$ with simple
objects $i$ satisfying
$f(-)+t(i,\bl a(-)) \in \Z$ [note that $K(x,-)$ is an integer for $x\in \ker\bl a$, so
this is well defined for $f+K(2\ker\bl a,-)$].
 Denote the simple objects in $\cC^\uparrow_{\bl a,K}$
by $i_{f}$. We then define the tensor
product and braiding in $\cC^\uparrow_{\bl a,K}$,
\begin{gather}
i_{f}\ot j_{g}=(i\ot j)_{f+g}=\oplus_k N^{ij}_k k_{f+g},\\
  \alpha_{i_{f},j_{g},k_{h}}=\alpha_{i,j,k},\\
  c_{i_{f},j_{g}}=c_{i,j}\ee^{\ii\pi \ov K(f,g)}.\label{eq.fc}
\end{gather}
\eqref{eq.fc} is independent of the choice of representative: $\forall x\in \ker \bl a,$
\begin{align}
  c_{i_{f+K(2x,-)},j_{g}}&=c_{i_f,j_g}\ee^{\ii\pi \ov K( K(2x,-),g)}\nonumber\\
  &=c_{i_f,j_g}\ee^{2\pi\ii g(x)}.
\end{align}
Since $t(j,\bl a(x))=t(j,0)\in \Z$, clearly $g(x)\in \Z$ as desired. Thus
$\cC^\uparrow_{\bl a,K}$ is a braided fusion category graded by $\ov\Z^\kappa/K(2\ker \bl a,-)$.
It is obvious that $d_{i_f}=d_i$.

Observe that for any $x\in \Z^\kappa$,
$\bl a(x)_{K(x,-)}$ is a self boson and mutually trivial to any object $i_f$.
$\bl a(x)_{K(x,-)}$ is a self boson since
\begin{align}
  &\Tr c_{\bl a(x)_{K(x,-)},\bl a(x)_{K(x,-)}}\nonumber\\
  &=\Tr c_{\bl a(x),\bl a(x)}\ee^{\ii\pi \ov
  K(K(x,-),K(x,-)) }\nonumber\\
  &=\ee^{\ii\pi [t(\bl a(x),\bl a(x))+ K(x,x)] }=1.
\end{align}
$\bl a(x)_{K(x,-)}$ is in the M\"uger center\cite{Muger0201017} (mutually trivial to any object
$i_f$) since
\begin{align}
  &\frac 1{d_i}\Tr c_{i_f,\bl a(x)_{K(x,-)}}c_{\bl a(x)_{K(x,-)},i_f}\nonumber\\
  &=\ee^{2\pi\ii [t(i,\bl a(x))+\ov K( f,K(x,-))]}\nonumber\\
&=\ee^{2\pi\ii [t(i,\bl a(x))+f(x))]}=1.
\end{align}
Therefore, $\{ \bl a(x)_{K(x,-)},\ x\in \Z^\kappa \}$ generates a symmetric fusion
subcategory in the M\"uger center of $\cC^\uparrow_{\bl a,K}$ which is equivalent to $\Rep(\<\bl a(x)_{K(x,-)}\>\simeq \Z^\kappa/2\ker \bl a)$.
Condense it~\cite{Kong1307.8244} (take the category of local modules over $\Fun( \Z^\kappa/2\ker \bl a)$), and we obtain the final result
$\cC_{\bl a,K}=(\cC^\uparrow_{\bl a,K})^{\mathrm{loc}}_{\Fun(\Z^\kappa/2\ker \bl a)}$.
In general the associator ($F$ matrix) will change and get complicated after
such anyon condensation process. However, since the condensed anyons are in the
M\"uger center, the braiding and fusion
rules are preserved~\cite{Kong1307.8244,DMNO1009.2117}. Thus if we are only interested in the simple data such
as fusion rules and $T,S$ matrices, it is fine to work in the larger category
$\cC^\uparrow_{\bl a,K}$.

\prlsec{Conclusion and outlook} In this letter we introduced the matrix
formulation for non-Abelian families, which makes it possible to generate any
topological order in the same non-Abelian family as a given one
almost instantly. We have provided a powerful tool, which, on one hand, can help
group known topological
orders~\cite{BBCW1410.4540,Wen1506.05768,LKW1507.04673,LKW1602.05936,LKW1602.05946}
(or modular tensor cagegories~\cite{RSW0712.1377}) into non-Abelian
families, and for simplicity, only the data of one root is necessary to be
listed explicitly; on the other hand, one can efficiently generate the data of
infinitely many possible unknown topological orders. 

The results in \Ref{LW1701.07820} already reduces the classification problem of
all 2+1D topological orders to the classification of all root topological
orders, namely in which the Abelian anyons have trivial self and mutual
statistics. The results in this letter further makes this reduction an efficient
and simple algorithm. In the end, we only need to maintain a list of root topological
orders. It will be interesting to find the canonical (the simplest) form of
$(\bl a, K)$, and then we will have a simple name for each topological
order: the root $\cC$ plus the canonical form of $(\bl a,K)$. Moreover, after fixing a root $\cC$, we should be able to extract all
possible \emph{non-Abelian invariants}~\cite{Lan2017} of this family by studying $\cC$ and
the pair $(\bl a,K)$. These non-Abelian invariants will surely deepen our
understanding on topological phases of matter, as well as on the application of topological materials in quantum computation.

Our construction can also be viewed as a generalization of anyon
condensation~\cite{Kong1307.8244}, where anyons are forced to form an effective
trivial state, and the condensed anyons are necessarily bosons. We make anyons
form effective Laughlin states, and our results imply that the
multilayer Laughlin states are the most general type of states Abelian anyons
can form. From this point of view, it is natural to ask what kind of
nontrivial effective states non-Abelian anyons can form. Furture research along
this line may reveal more exotic relations between topological phases, by
nontrivial condensations of non-Abelian anyons, and further simplify our
understanding of topological orders.

TL thanks Zhihao Zhang and Wenjie Xi for helpful
discussions. This work was done during TL's visit at Center for Quantum
Computing, Peng Cheng Laboratory and Southern University of Science and
Technology.
\bibliography{../library} 

\begin{thebibliography}{20}%
\makeatletter
\providecommand \@ifxundefined [1]{%
 \@ifx{#1\undefined}
}%
\providecommand \@ifnum [1]{%
 \ifnum #1\expandafter \@firstoftwo
 \else \expandafter \@secondoftwo
 \fi
}%
\providecommand \@ifx [1]{%
 \ifx #1\expandafter \@firstoftwo
 \else \expandafter \@secondoftwo
 \fi
}%
\providecommand \natexlab [1]{#1}%
\providecommand \enquote  [1]{``#1''}%
\providecommand \bibnamefont  [1]{#1}%
\providecommand \bibfnamefont [1]{#1}%
\providecommand \citenamefont [1]{#1}%
\providecommand \href@noop [0]{\@secondoftwo}%
\providecommand \href [0]{\begingroup \@sanitize@url \@href}%
\providecommand \@href[1]{\@@startlink{#1}\@@href}%
\providecommand \@@href[1]{\endgroup#1\@@endlink}%
\providecommand \@sanitize@url [0]{\catcode `\\12\catcode `\$12\catcode
  `\&12\catcode `\#12\catcode `\^12\catcode `\_12\catcode `\%12\relax}%
\providecommand \@@startlink[1]{}%
\providecommand \@@endlink[0]{}%
\providecommand \url  [0]{\begingroup\@sanitize@url \@url }%
\providecommand \@url [1]{\endgroup\@href {#1}{\urlprefix }}%
\providecommand \urlprefix  [0]{URL }%
\providecommand \Eprint [0]{\href }%
\providecommand \doibase [0]{http://dx.doi.org/}%
\providecommand \selectlanguage [0]{\@gobble}%
\providecommand \bibinfo  [0]{\@secondoftwo}%
\providecommand \bibfield  [0]{\@secondoftwo}%
\providecommand \translation [1]{[#1]}%
\providecommand \BibitemOpen [0]{}%
\providecommand \bibitemStop [0]{}%
\providecommand \bibitemNoStop [0]{.\EOS\space}%
\providecommand \EOS [0]{\spacefactor3000\relax}%
\providecommand \BibitemShut  [1]{\csname bibitem#1\endcsname}%
\let\auto@bib@innerbib\@empty
\bibitem [{\citenamefont {Wen}(1989)}]{Wen89}%
  \BibitemOpen
  \bibfield  {author} {\bibinfo {author} {\bibfnamefont {X.~G.}\ \bibnamefont
  {Wen}},\ }\href {\doibase 10.1103/PhysRevB.40.7387} {\bibfield  {journal}
  {\bibinfo  {journal} {Physical Review B}\ }\textbf {\bibinfo {volume} {40}},\
  \bibinfo {pages} {7387} (\bibinfo {year} {1989})}\BibitemShut {NoStop}%
\bibitem [{\citenamefont {Wen}(1990)}]{Wen90}%
  \BibitemOpen
  \bibfield  {author} {\bibinfo {author} {\bibfnamefont {X.~G.}\ \bibnamefont
  {Wen}},\ }\href {\doibase 10.1142/S0217979290000139} {\bibfield  {journal}
  {\bibinfo  {journal} {International Journal of Modern Physics B}\ }\textbf
  {\bibinfo {volume} {04}},\ \bibinfo {pages} {239} (\bibinfo {year}
  {1990})}\BibitemShut {NoStop}%
\bibitem [{\citenamefont {Wen}\ and\ \citenamefont {Niu}(1990)}]{WN90}%
  \BibitemOpen
  \bibfield  {author} {\bibinfo {author} {\bibfnamefont {X.~G.}\ \bibnamefont
  {Wen}}\ and\ \bibinfo {author} {\bibfnamefont {Q.}~\bibnamefont {Niu}},\
  }\href {\doibase 10.1103/PhysRevB.41.9377} {\bibfield  {journal} {\bibinfo
  {journal} {Physical Review B}\ }\textbf {\bibinfo {volume} {41}},\ \bibinfo
  {pages} {9377} (\bibinfo {year} {1990})}\BibitemShut {NoStop}%
\bibitem [{\citenamefont {Kitaev}(2003)}]{Kitaev9707021}%
  \BibitemOpen
  \bibfield  {author} {\bibinfo {author} {\bibfnamefont {A.}~\bibnamefont
  {Kitaev}},\ }\href {\doibase 10.1016/S0003-4916(02)00018-0} {\bibfield
  {journal} {\bibinfo  {journal} {Annals of Physics}\ }\textbf {\bibinfo
  {volume} {303}},\ \bibinfo {pages} {2} (\bibinfo {year} {2003})},\ \Eprint
  {http://arxiv.org/abs/9707021} {arXiv:9707021 [quant-ph]} \BibitemShut
  {NoStop}%
\bibitem [{\citenamefont {Freedman}\ \emph {et~al.}(2002)\citenamefont
  {Freedman}, \citenamefont {Kitaev}, \citenamefont {Larsen},\ and\
  \citenamefont {Wang}}]{FKLW01}%
  \BibitemOpen
  \bibfield  {author} {\bibinfo {author} {\bibfnamefont {M.~H.}\ \bibnamefont
  {Freedman}}, \bibinfo {author} {\bibfnamefont {A.}~\bibnamefont {Kitaev}},
  \bibinfo {author} {\bibfnamefont {M.~J.}\ \bibnamefont {Larsen}}, \ and\
  \bibinfo {author} {\bibfnamefont {Z.}~\bibnamefont {Wang}},\ }\href {\doibase
  10.1090/S0273-0979-02-00964-3} {\bibfield  {journal} {\bibinfo  {journal}
  {Bulletin of the American Mathematical Society}\ }\textbf {\bibinfo {volume}
  {40}},\ \bibinfo {pages} {31} (\bibinfo {year} {2002})},\ \Eprint
  {http://arxiv.org/abs/0101025} {arXiv:0101025 [quant-ph]} \BibitemShut
  {NoStop}%
\bibitem [{\citenamefont {Lan}\ and\ \citenamefont {Wen}(2017)}]{LW1701.07820}%
  \BibitemOpen
  \bibfield  {author} {\bibinfo {author} {\bibfnamefont {T.}~\bibnamefont
  {Lan}}\ and\ \bibinfo {author} {\bibfnamefont {X.-G.}\ \bibnamefont {Wen}},\
  }\href {\doibase 10.1103/PhysRevLett.119.040403} {\bibfield  {journal}
  {\bibinfo  {journal} {Physical Review Letters}\ }\textbf {\bibinfo {volume}
  {119}},\ \bibinfo {pages} {040403} (\bibinfo {year} {2017})},\ \Eprint
  {http://arxiv.org/abs/1701.07820} {arXiv:1701.07820} \BibitemShut {NoStop}%
\bibitem [{\citenamefont {Wen}\ and\ \citenamefont {Zee}(1992)}]{WZ92}%
  \BibitemOpen
  \bibfield  {author} {\bibinfo {author} {\bibfnamefont {X.~G.}\ \bibnamefont
  {Wen}}\ and\ \bibinfo {author} {\bibfnamefont {A.}~\bibnamefont {Zee}},\
  }\href {\doibase 10.1103/PhysRevB.46.2290} {\bibfield  {journal} {\bibinfo
  {journal} {Physical Review B}\ }\textbf {\bibinfo {volume} {46}},\ \bibinfo
  {pages} {2290} (\bibinfo {year} {1992})}\BibitemShut {NoStop}%
\bibitem [{\citenamefont {Laughlin}(1983)}]{Laughlin1983}%
  \BibitemOpen
  \bibfield  {author} {\bibinfo {author} {\bibfnamefont {R.~B.}\ \bibnamefont
  {Laughlin}},\ }\href {\doibase 10.1103/PhysRevLett.50.1395} {\bibfield
  {journal} {\bibinfo  {journal} {Physical Review Letters}\ }\textbf {\bibinfo
  {volume} {50}},\ \bibinfo {pages} {1395} (\bibinfo {year}
  {1983})}\BibitemShut {NoStop}%
\bibitem [{\citenamefont {Haldane}(1983)}]{Haldane83}%
  \BibitemOpen
  \bibfield  {author} {\bibinfo {author} {\bibfnamefont {F.~D.~M.}\
  \bibnamefont {Haldane}},\ }\href {\doibase 10.1103/PhysRevLett.51.605}
  {\bibfield  {journal} {\bibinfo  {journal} {Physical Review Letters}\
  }\textbf {\bibinfo {volume} {51}},\ \bibinfo {pages} {605} (\bibinfo {year}
  {1983})}\BibitemShut {NoStop}%
\bibitem [{\citenamefont {Halperin}(1984)}]{Halperin84}%
  \BibitemOpen
  \bibfield  {author} {\bibinfo {author} {\bibfnamefont {B.~I.}\ \bibnamefont
  {Halperin}},\ }\href {\doibase 10.1103/PhysRevLett.52.1583} {\bibfield
  {journal} {\bibinfo  {journal} {Physical Review Letters}\ }\textbf {\bibinfo
  {volume} {52}},\ \bibinfo {pages} {1583} (\bibinfo {year}
  {1984})}\BibitemShut {NoStop}%
\bibitem [{\citenamefont {Lan}(2017)}]{Lan2017}%
  \BibitemOpen
  \bibfield  {author} {\bibinfo {author} {\bibfnamefont {T.}~\bibnamefont
  {Lan}},\ }\emph {\bibinfo {title} {{A Classification of (2+1)D Topological
  Phases with Symmetries}}},\ \href
  {https://uwspace.uwaterloo.ca/handle/10012/12389} {Ph.D. thesis},\ \bibinfo
  {school} {University of Waterloo} (\bibinfo {year} {2017})\BibitemShut
  {NoStop}%
\bibitem [{\citenamefont {M{\"{u}}ger}(2003)}]{Muger0201017}%
  \BibitemOpen
  \bibfield  {author} {\bibinfo {author} {\bibfnamefont {M.}~\bibnamefont
  {M{\"{u}}ger}},\ }\href {\doibase 10.1112/S0024611503014187} {\bibfield
  {journal} {\bibinfo  {journal} {Proceedings of the London Mathematical
  Society}\ }\textbf {\bibinfo {volume} {87}},\ \bibinfo {pages} {291}
  (\bibinfo {year} {2003})},\ \Eprint {http://arxiv.org/abs/0201017}
  {arXiv:0201017 [math]} \BibitemShut {NoStop}%
\bibitem [{\citenamefont {Kong}(2014)}]{Kong1307.8244}%
  \BibitemOpen
  \bibfield  {author} {\bibinfo {author} {\bibfnamefont {L.}~\bibnamefont
  {Kong}},\ }\href {\doibase 10.1016/j.nuclphysb.2014.07.003} {\bibfield
  {journal} {\bibinfo  {journal} {Nuclear Physics B}\ }\textbf {\bibinfo
  {volume} {886}},\ \bibinfo {pages} {436} (\bibinfo {year} {2014})},\ \Eprint
  {http://arxiv.org/abs/1307.8244} {arXiv:1307.8244} \BibitemShut {NoStop}%
\bibitem [{\citenamefont {Davydov}\ \emph {et~al.}(2013)\citenamefont
  {Davydov}, \citenamefont {M{\"{u}}ger}, \citenamefont {Nikshych},\ and\
  \citenamefont {Ostrik}}]{DMNO1009.2117}%
  \BibitemOpen
  \bibfield  {author} {\bibinfo {author} {\bibfnamefont {A.}~\bibnamefont
  {Davydov}}, \bibinfo {author} {\bibfnamefont {M.}~\bibnamefont
  {M{\"{u}}ger}}, \bibinfo {author} {\bibfnamefont {D.}~\bibnamefont
  {Nikshych}}, \ and\ \bibinfo {author} {\bibfnamefont {V.}~\bibnamefont
  {Ostrik}},\ }\href {\doibase 10.1515/crelle.2012.014} {\bibfield  {journal}
  {\bibinfo  {journal} {Journal f{\"{u}}r die reine und angewandte Mathematik
  (Crelles Journal)}\ }\textbf {\bibinfo {volume} {2013}},\ \bibinfo {pages}
  {135} (\bibinfo {year} {2013})},\ \Eprint {http://arxiv.org/abs/1009.2117}
  {arXiv:1009.2117} \BibitemShut {NoStop}%
\bibitem [{\citenamefont {Barkeshli}\ \emph {et~al.}(2014)\citenamefont
  {Barkeshli}, \citenamefont {Bonderson}, \citenamefont {Cheng},\ and\
  \citenamefont {Wang}}]{BBCW1410.4540}%
  \BibitemOpen
  \bibfield  {author} {\bibinfo {author} {\bibfnamefont {M.}~\bibnamefont
  {Barkeshli}}, \bibinfo {author} {\bibfnamefont {P.}~\bibnamefont
  {Bonderson}}, \bibinfo {author} {\bibfnamefont {M.}~\bibnamefont {Cheng}}, \
  and\ \bibinfo {author} {\bibfnamefont {Z.}~\bibnamefont {Wang}},\ }\href
  {http://arxiv.org/abs/1410.4540} {\bibfield  {journal} {\bibinfo  {journal}
  {ArXiv e-prints}\ } (\bibinfo {year} {2014})},\ \Eprint
  {http://arxiv.org/abs/1410.4540} {arXiv:1410.4540} \BibitemShut {NoStop}%
\bibitem [{\citenamefont {Wen}(2016)}]{Wen1506.05768}%
  \BibitemOpen
  \bibfield  {author} {\bibinfo {author} {\bibfnamefont {X.-G.}\ \bibnamefont
  {Wen}},\ }\href {\doibase 10.1093/nsr/nwv077} {\bibfield  {journal} {\bibinfo
   {journal} {National Science Review}\ }\textbf {\bibinfo {volume} {3}},\
  \bibinfo {pages} {68} (\bibinfo {year} {2016})},\ \Eprint
  {http://arxiv.org/abs/1506.05768} {arXiv:1506.05768} \BibitemShut {NoStop}%
\bibitem [{\citenamefont {Lan}\ \emph {et~al.}(2016)\citenamefont {Lan},
  \citenamefont {Kong},\ and\ \citenamefont {Wen}}]{LKW1507.04673}%
  \BibitemOpen
  \bibfield  {author} {\bibinfo {author} {\bibfnamefont {T.}~\bibnamefont
  {Lan}}, \bibinfo {author} {\bibfnamefont {L.}~\bibnamefont {Kong}}, \ and\
  \bibinfo {author} {\bibfnamefont {X.-G.}\ \bibnamefont {Wen}},\ }\href
  {\doibase 10.1103/PhysRevB.94.155113} {\bibfield  {journal} {\bibinfo
  {journal} {Physical Review B}\ }\textbf {\bibinfo {volume} {94}},\ \bibinfo
  {pages} {155113} (\bibinfo {year} {2016})},\ \Eprint
  {http://arxiv.org/abs/1507.04673} {arXiv:1507.04673} \BibitemShut {NoStop}%
\bibitem [{\citenamefont {Lan}\ \emph {et~al.}(2017{\natexlab{a}})\citenamefont
  {Lan}, \citenamefont {Kong},\ and\ \citenamefont {Wen}}]{LKW1602.05936}%
  \BibitemOpen
  \bibfield  {author} {\bibinfo {author} {\bibfnamefont {T.}~\bibnamefont
  {Lan}}, \bibinfo {author} {\bibfnamefont {L.}~\bibnamefont {Kong}}, \ and\
  \bibinfo {author} {\bibfnamefont {X.-G.}\ \bibnamefont {Wen}},\ }\href
  {\doibase 10.1007/s00220-016-2748-y} {\bibfield  {journal} {\bibinfo
  {journal} {Communications in Mathematical Physics}\ }\textbf {\bibinfo
  {volume} {351}},\ \bibinfo {pages} {709} (\bibinfo {year}
  {2017}{\natexlab{a}})},\ \Eprint {http://arxiv.org/abs/1602.05936}
  {arXiv:1602.05936} \BibitemShut {NoStop}%
\bibitem [{\citenamefont {Lan}\ \emph {et~al.}(2017{\natexlab{b}})\citenamefont
  {Lan}, \citenamefont {Kong},\ and\ \citenamefont {Wen}}]{LKW1602.05946}%
  \BibitemOpen
  \bibfield  {author} {\bibinfo {author} {\bibfnamefont {T.}~\bibnamefont
  {Lan}}, \bibinfo {author} {\bibfnamefont {L.}~\bibnamefont {Kong}}, \ and\
  \bibinfo {author} {\bibfnamefont {X.-G.}\ \bibnamefont {Wen}},\ }\href
  {\doibase 10.1103/PhysRevB.95.235140} {\bibfield  {journal} {\bibinfo
  {journal} {Physical Review B}\ }\textbf {\bibinfo {volume} {95}},\ \bibinfo
  {pages} {235140} (\bibinfo {year} {2017}{\natexlab{b}})},\ \Eprint
  {http://arxiv.org/abs/1602.05946} {arXiv:1602.05946} \BibitemShut {NoStop}%
\bibitem [{\citenamefont {Rowell}\ \emph {et~al.}(2009)\citenamefont {Rowell},
  \citenamefont {Stong},\ and\ \citenamefont {Wang}}]{RSW0712.1377}%
  \BibitemOpen
  \bibfield  {author} {\bibinfo {author} {\bibfnamefont {E.}~\bibnamefont
  {Rowell}}, \bibinfo {author} {\bibfnamefont {R.}~\bibnamefont {Stong}}, \
  and\ \bibinfo {author} {\bibfnamefont {Z.}~\bibnamefont {Wang}},\ }\href
  {\doibase 10.1007/s00220-009-0908-z} {\bibfield  {journal} {\bibinfo
  {journal} {Communications in Mathematical Physics}\ }\textbf {\bibinfo
  {volume} {292}},\ \bibinfo {pages} {343} (\bibinfo {year} {2009})},\ \Eprint
  {http://arxiv.org/abs/0712.1377} {arXiv:0712.1377} \BibitemShut {NoStop}%
\end{thebibliography}%
\appendix
\section{Proof of Theorem~\ref{thm.main}}\label{ap.proof}
  We prove the theorem by induction. It is obviously true for $\kappa=1$. Now
  assume that it is true for $\kappa-1$ where $\kappa>1$. Let $K_0$ be the corresponding
  $\kappa-1$ by
  $\kappa-1$ matrix. From $\kappa-1$ to $\kappa$ we choose $a_c=(a_{\kappa},\bl
  l_c)$ and even integer $m_c$. The new $K$ matrix is
\begin{align}
  K_1= \bpm K_0 & \bl l_c-\bl t_{a_\kappa} \\
  (\bl l_c-\bl t_{a_\kappa})^T &  m_c-2s_{a_\kappa} \\
     \epm.
\end{align}
The spin of $a_c$ is
\begin{align}
  s_{a_c}=s_{a_{\kappa}}+\frac12 (\bl l_c-\bl t_{a_\kappa})^T K_0^{-1} (\bl l_c-\bl
  t_{a_\kappa}),
\end{align}
and the mutual statistics between $(i,\bl l_0)$ and $a_c$ is
\begin{align}
t_{(i,\bl l_0)} = t_{i,a_\kappa}+ (\bl l_c-\bl t_{a_\kappa})^T K_0^{-1} (\bl
l_0-\bl t_i).
\end{align}

First, as long as $m_c-2s_{a_c}\neq 0$, $K_1$ is invertible with
\begin{align}
 K_1^{-1} = 
\bpm
 K_0^{-1} +\frac{ K_0^{-1} (\bl l_c-\bl t_{a_\kappa}) (\bl l_c-\bl t_{a_\kappa})^T K_0^{-1}}{ m_c - 2s_{a_c}}
& - \frac{K_0^{-1} (\bl l_c-\bl t_{a_\kappa})}{ m_c - 2s_{a_c}} \\
- \frac{(\bl l_c-\bl t_{a_\kappa})^T K_0^{-1} }{ m_c - 2s_{a_c}} &  \frac{1}{ m_c - 2s_{a_c}}
\epm
\end{align}
Also 
\begin{align}
  \det(K_1)&=\det\bpm K_0 & \bl l_c-\bl t_{a_\kappa} \\
  (\bl l_c-\bl t_{a_\kappa})^T &  m_c-2s_{a_\kappa} \\
  \epm \nonumber\\
  &= \det(K_0)(m_c-2s_{a_\kappa}-(\bl l_c-\bl t_{a_\kappa})^T K_0^{-1} (\bl l_c-\bl
  t_{a_\kappa}))
     \nonumber\\
     &=(m_c-2s_{a_c})\det(K_0).
\end{align}
Thus $\det K$ accounts for the increment of rank, total quantum dimension, as
well as the normalization of $S$ matrix.
Also $\sgn K=\sgn K_0+ \sgn (m_c-2 s_{a_c})$ accounts for the increment of chiral central charge.

The new anyons are labeled by $(i,\bl l_0,m)$ where $m$ is an integer. Combine
$\bl l$ and $m$ into a $\kappa$-dimensional vector $\bl l^T=(\bl l_0^T,m)$. We
only need to verify the spin, equivalence relations and fusion rule of $(i,\bl
l)$; $S$ matrix follows directly.

The spin of
$(i,\bl l_0,m)=(i,\bl l)$ is
\begin{align}
  s_{(i,\bl l)}  &= s_{(i,\bl l_0)}+\frac{(m-t_{(i,\bl
  l_0)})^2}{2(m_c-2s_{a_c})}\nonumber\\
&=s_i+\frac12 (\bl l_0-\bl t_i)^T K_0^{-1} (\bl l_0-\bl t_i)+\frac{(m-t_{(i,\bl
l_0)})^2}{2(m_c-2s_{a_c})}.
\end{align}
While (using the same notation for $\kappa-1$ and $\kappa$ dimensional $\bl
t_i$) 
\begin{align}
  &\frac12 (\bl l-\bl t_i)^T K_1^{-1} (\bl l-\bl t_{i})\nonumber\\
  &=\frac12 \left((\bl l_0-\bl t_i)^T,m-t_{i,a_{\kappa}}\right) K_1^{-1} \left(
  \begin{array}{c}
    \bl l_0-\bl t_{i}\\ m-t_{i,a_\kappa}
  \end{array}
\right)\nonumber\\
 &= \frac12 \Big((\bl l_0-\bl t_i)^T K_0^{-1} (\bl l_0-\bl t_i) 
   \nonumber\\
   &+
 \frac{(m-t_{i,a_\kappa})^2+(t_{(i,\bl l_0)}-t_{i,a_\kappa})^2 -2
 (m-t_{i,a_\kappa}) (t_{(i,\bl l_0)}-t_{i,a_\kappa}) }{m_c-2s_{a_c}}  \Big) .
\end{align}
Indeed we have
\begin{align}
	s_{(i,\bl l)}=s_i+\frac12 (\bl l-\bl t_i)^T K_1^{-1} (\bl l-\bl t_i). 
\end{align}
For $\kappa-1$ we have equivalence relations
\begin{align}
(i,\bl l_0)\sim (i\ot a_I,\bl l_0 +(K_0)_I-\bl t_i+\bl t_{i\ot a_I}).
\end{align}
For $(i,\bl l_0,m)$, one equivalence relation comes from condensing
$a_c=(a_\kappa,\bl l_c)$ with even integer $m_c$,
\begin{widetext}
\begin{align}
  (i,\bl l_0,m)\sim (i\ot a_\kappa, \bl l_0+\bl l_c -\bl t_i-\bl t_{a_\kappa}
    +\bl t_{i\ot a_\kappa},
    m+m_c-t_{(i,\bl l_0)}-t_{(a_\kappa,\bl l_c)}+t_{(i\ot a_\kappa, \bl l_0+\bl l_c -\bl t_i-\bl t_{a_\kappa}
  +\bl t_{i\ot a_\kappa})}),
\end{align}
where
\begin{align}
  &-t_{(i,\bl l_0)}-t_{(a_\kappa,\bl l_c)}+t_{(i\ot a_\kappa, \bl l_0+\bl l_c -\bl t_i-\bl t_{a_\kappa}
  +\bl t_{i\ot a_\kappa})}\nonumber\\
  &=-t_{i,a_\kappa}-t_{a_\kappa,a_\kappa}+t_{i\ot a_\kappa,a_\kappa}
  +(\bl l_c-\bl t_{a_\kappa})^T K_0^{-1}(\bl t_i-\bl l_0+ \bl t_{a_\kappa}-\bl
    l_c- \bl t_{i\ot a_\kappa}+ \bl l_0+\bl l_c -\bl t_i-\bl t_{a_\kappa}
  +\bl t_{i\ot a_\kappa})\nonumber\\
  &=-t_{i,a_\kappa}-t_{a_\kappa,a_\kappa}+t_{i\ot a_\kappa,a_\kappa}.
\end{align}
Thus
\begin{align}
  (i,\bl l)\sim (i\ot a_\kappa, \bl l + (K_1)_\kappa- \bl t_i+\bl t_{i\ot
  a_\kappa}),
\end{align}
where $(K_1)_\kappa^T= (\bl l_c^T-\bl t_{a_\kappa}^T,m_c-2s_{a_\kappa}).$
The other equivalence relations come from choosing a different representative of
$(i,\bl l_0)$; for $I=1,\dots,\kappa-1$,
\begin{align}
  (i,\bl l_0,m)\sim (i\ot a_I,\bl l_0+(K_0)_I-\bl t_i+\bl t_{i\ot a_I}, m-t_{(i,\bl
  l_0)}+t_{(i\ot a_I,\bl l_0+K_I-\bl t_i+\bl t_{i\ot a_I})}),
\end{align}
where
\begin{align}
  &-t_{(i,\bl
  l_0)}+t_{(i\ot a_I,\bl l_0+(K_0)_I-\bl t_i+\bl t_{i\ot a_I})}\nonumber\\
  &=-t_{i,a_\kappa}+t_{i\ot a_I,a_\kappa}+(\bl l_c-\bl t_{a_\kappa})^T
K_0^{-1}(\bl t_i-\bl l_0-\bl t_{i\ot a_I}+\bl l_0+(K_0)_I-\bl t_i+\bl
t_{i\ot a_I})\nonumber\\
&=-t_{i,a_\kappa}+t_{i\ot a_\kappa,a_\kappa}+(\bl l_c-\bl t_{a_\kappa})^T_I.
\end{align}
Thus
\begin{align}
  (i,\bl l)\sim (i\ot a_I,\bl l+ (K_1)_I-\bl t_i+\bl t_{i\ot a_I}),
\end{align}
where $(K_1)_I^T=((K_0)_I^T,(\bl l_c-\bl t_{a_\kappa})_I)$,
$I=1,\dots,\kappa-1$.

The fusion of $(i,\bl l_0,m)$ and $(j,\bl k_0,n)$ is
\begin{align}
  (i,\bl l_0,m)\ot (j,\bl k_0,n)=\oplus N^{ij}_s(s,\bl l_0+\bl k_0 -\bl t_i-
  \bl t_j +\bl t_s, m+n-t_{(i,\bl l_0)}-t_{(j,\bl k_0)}+ t_{(s,\bl l_0+\bl k_0 -\bl t_i-
\bl t_j +\bl t_s)}),
\end{align}
where
\begin{align}
  &-t_{(i,\bl l_0)}-t_{(j,\bl k_0)}+ t_{(s,\bl l_0+\bl k_0 -\bl t_i-
\bl t_j +\bl t_s)}\nonumber\\
&=-t_{i,a_\kappa}-t_{j,a_\kappa}+t_{s,a_\kappa}+ (\bl l_c-\bl t_{a_\kappa})^T
K_0^{-1}(\bl t_i-\bl l_0+\bl t_j-\bl k_0-\bl t_s+\bl l_0+\bl k_0 -\bl t_i-
\bl t_j +\bl t_s)\nonumber\\
&=-t_{i,a_\kappa}-t_{j,a_\kappa}+t_{s,a_\kappa}.
\end{align}
\end{widetext}
Thus we do have
\begin{align}
  (i,\bl l)\ot (j,\bl k)=\oplus N^{ij}_s(s,\bl l+\bl k-\bl t_i-\bl t_j +\bl
  t_s).
\end{align}
In the above proof, we need to assume that $\det K_0\neq 0$. As we prove by induction, this in fact means that we need to assume that $\det (K_{IJ},\ I,J=1,2,\dots, n)\neq 0$ for any $n<\kappa-1$. However, such assumption is inessential and can be dropped, given the following transformation on $(\bl a,K)$: 
    for an integer
      matrix $W$ with $\det W=\pm 1$, $(a_I,K_{IJ})$ is equivalent to
      $(a'_I=\ot_J a_J^{\ot W_{IJ}}, K'= WKW^T)$. The fact that
      $t_{a'_I,a'_J}=\sum_{PQ} W_{IP} t_{a_P,a_Q} W_{JQ}$ implies the
      transformation for the $K$ matrix.
      As $a_I$ are in an Abelian group, it is convenient to write in the additive
      convention $a'_I=\sum_J W_{IJ} a_I$, or simply $\bl a'=W \bl a$.
      Thus $\cC_{\bl a,K}\simeq \cC_{W\bl a, WKW^T}$ for integer matrix $W$ with
      $\det W=\pm 1$. More precisely, the equivalence is given by $(i,\bl l)\mapsto (i,W\bl l)$. Note that $\bl t'_i=(t_{i,a'_I})=W\bl t_i$. It is straightforward to check that this map is compatible with the equivalence relation \eqref{ilequiv}, and preserves fusion \eqref{ilfuse}, spin \eqref{ilspin}, and $S$ matrix \eqref{ilmutual}.
\end{document}